\documentclass[twocolumn,journal]{IEEEtran}

\usepackage{amsfonts}
\usepackage{amsmath}
\usepackage{amsthm}
\usepackage{amssymb}
\usepackage{graphicx}
\usepackage[T1]{fontenc}
\usepackage{supertabular}
\usepackage{longtable}
\usepackage[usenames,dvipsnames]{color}
\usepackage{bbm}
\usepackage{fancyhdr}
\usepackage{breqn}

\usepackage{capt-of}
\usepackage{tikz}
\usetikzlibrary{matrix}
\usepackage{endnotes}
\usepackage{soul}
\usepackage{marginnote}
\newtheorem{example}{Example}

\newtheorem{commentary}{Comment}
\newcommand{\mathsym}[1]{}
\newcommand{\unicode}[1]{}

\usepackage{colortbl}

\usepackage[framemethod=TikZ]{mdframed}
\usepackage[framemethod=TikZ]{mdframed}
\usepackage{framed}

\mdfsetup{%
	outerlinewidth=1,skipabove=20pt,backgroundcolor=yellow!30, outerlinecolor=black,innertopmargin=0pt,splittopskip=\topskip,skipbelow=\baselineskip, skipabove=\baselineskip,ntheorem,roundcorner=5pt}

\mdtheorem[nobreak=true,outerlinewidth=1,
backgroundcolor=yellow!50, outerlinecolor=black,innertopmargin=0pt,splittopskip=\topskip,skipbelow=\baselineskip, skipabove=\baselineskip,ntheorem,roundcorner=5pt,font=\itshape]{result}{Result}

\mdtheorem[nobreak=true,outerlinewidth=1,
backgroundcolor=yellow!30, outerlinecolor=black,innertopmargin=0pt,splittopskip=\topskip,skipbelow=\baselineskip, skipabove=\baselineskip,ntheorem,roundcorner=5pt,font=\itshape]{theorem}{Theorem}

\mdtheorem[nobreak=true,outerlinewidth=1,
backgroundcolor=gray!10, outerlinecolor=black,innertopmargin=0pt,splittopskip=\topskip,skipbelow=\baselineskip, skipabove=\baselineskip,ntheorem,roundcorner=5pt,font=\itshape]{remark}{Remark}

\mdtheorem[nobreak=true,outerlinewidth=1,
backgroundcolor=gray!10, outerlinecolor=gray!10,innertopmargin=0pt,splittopskip=\topskip,skipbelow=\baselineskip, skipabove=\baselineskip,ntheorem,roundcorner=5pt,font=\itshape]{definition}{Definition}

\mdtheorem[nobreak=true,outerlinewidth=1,
backgroundcolor=pink!30, outerlinecolor=black,innertopmargin=0pt,splittopskip=\topskip,skipbelow=\baselineskip, skipabove=\baselineskip,ntheorem,roundcorner=5pt,font=\itshape]{quaestio}{Quaestio}

\mdtheorem[nobreak=true,outerlinewidth=1,
backgroundcolor=yellow!50, outerlinecolor=black,innertopmargin=5pt,splittopskip=\topskip,skipbelow=\baselineskip, skipabove=\baselineskip,ntheorem,roundcorner=5pt,font=\itshape]{background}{Background}

\mdtheorem[nobreak=true,outerlinewidth=1,
backgroundcolor=gray!10, outerlinecolor=black,innertopmargin=5pt,splittopskip=\topskip,skipbelow=\baselineskip, skipabove=\baselineskip,ntheorem,roundcorner=5pt,font=\itshape]{nothing}{}

\mdtheorem[nobreak=true,outerlinewidth=1,
backgroundcolor=pink!50, outerlinecolor=black,innertopmargin=5pt,splittopskip=\topskip,skipbelow=\baselineskip, skipabove=\baselineskip,ntheorem,roundcorner=5pt,font=\itshape]{point}{Point}
\mdtheorem[nobreak=true,outerlinewidth=1,
backgroundcolor=pink!50, outerlinecolor=black,innertopmargin=5pt,splittopskip=\topskip,skipbelow=\baselineskip, skipabove=\baselineskip,ntheorem,roundcorner=5pt,font=\itshape]{lemma}{Lemma}

\mdtheorem[nobreak=true,outerlinewidth=1,
backgroundcolor=pink!50, outerlinecolor=black,innertopmargin=5pt,splittopskip=\topskip,skipbelow=\baselineskip, skipabove=\baselineskip,ntheorem,roundcorner=5pt,font=\itshape]{commentary}{Comment}
\begin{document}

\title{{\color{Red}
On the Statistical Differences between Binary Forecasts and Real World Payoffs}}

\author{Nassim Nicholas Taleb\IEEEauthorrefmark{1}\IEEEauthorrefmark{2}\\     
   \IEEEauthorblockA{  \IEEEauthorrefmark{1} NYU Tandon School of Engineering  \IEEEauthorrefmark{2} Universa Investments  \\
   Forthcoming, \textit{International Journal of Forecasting, 2020}}

   \thanks{\color{Brown}December 4, 2019. A version of this article was presented at the M4 competition conference in New York in October 2018. The author thanks Marcos Carreira, Peter Carr, Pasquale Cirillo, Zhuo Xi, two invaluable and anonymous referees, and, of course, the very, very patient Spyros Makridakis. A particular thanks to Raphael Douady who took the thankless task of reviewing the mathematical definitions and derivations.    }}

\maketitle



\flushbottom 

%

\section{Introduction/Abstract}
There can be considerable mathematical and statistical differences between the following two items:
\begin{enumerate}
 \item	(univariate) binary predictions, bets and "beliefs" (expressed as a specific "event" will happen/will not happen) and, on the other,
 \item real-world continuous payoffs (that is, numerical benefits or harm from an event).
 \end{enumerate}
 Way too often, the decision science and economics literature uses one as a proxy for another. Some results, say overestimation of tail \textit{probability}, by humans can be stated in one result\footnote{The notion of "calibration" is presented and explained in Fig. \ref{calibration} and its captions; for its origin, see Lichtenstein et al (\cite{lichtenstein1977calibration}, and more explicitly \cite{lichtenstein1978judged}  "The primary bias is the "overestimation of low frequencies and underestimation of ...high frequencies") and  most notably since the influential Kahneman and Tversky paper on prospect theory \cite{kahneman1979prospect}). Miscalibration and the point of the article are shown in Fig. \ref{miscalibration}.} and unwarranted conclusions that people overestimate \textit{tail risk} have been chronically made since.\footnote{See further \cite{johnson1983affect} , A more recent instance of such conflation is in a review paper by Barberis \cite{barberis2013psychology}).}\footnote{Some alternative mechanisms in say \cite{hertwig2005judgments} and more generally Gigerenzer's group, see a summary in \cite{gigerenzer2004fast}.}
 
  In this paper we show the mischaracterization as made in the decision-science literature and presents the effect of their conflation.  We also examine the differences under thin and fat tails --for under Gaussian distributions the effect can be marginal, which may have lulled the psychology literature into the conflation. 
 
 The net effects are:
 
\subsubsection{\textbf {Spuriousness of many psychological results}}
This affects risk management claims, particularly the research results to the effect that humans overestimate the risks of rare events.  Many perceived "biases" are shown to be just mischaracterizations by psychologists. We quantify such conflations with a metric for "pseudo-overestimation".

\subsubsection{\textbf {Being a "good forecaster" in binary space doesn't lead to having a good actual performance} }
The reverse is also true, and the effect  is exacerbated under nonlinearities. A binary forecasting record is likely to be a reverse indicator under some classes of distributions or deeper uncertainty. 

\subsubsection{ \textbf {Machine Learning}} Some nonlinear payoff functions, while not lending themselves to verbalistic expressions and "forecasts", can be well captured by ML or expressed in option contracts.

\subsubsection{\textbf { Fattailedness}} The difference is exaggerated  when the variable under consideration lies in the power law classes of probability distributions.

\subsubsection{\textbf {  Model error}} Binary forecasts are not particularly prone to model error; real world payoffs are.

\bigskip
The paper is organized as follows. We first present the difference in statistical properties thanks to precise mathematical definitions of the two types in section \ref{contdiscrete}. The text is structured with (numbered) "definitions","comments", and "examples".   Section \ref{nocollapse} presents the differences in the context of Gaussian-like and fat tailed environments (that is, the class of distributions dominated by remove events), a separation based on the presence or absence of a characteristic scale.  Section \ref{spurious} develops the mathematics of spurious overestimation, comparing the properties of payoffs under thin tails (section A) and Fat Tails (section B), discusses the conflation and presents the impact of model error (section D). Section \ref{miscalibr} applies to the calibration in psychological experiments. Section \ref{metrics} presents a catalogue of scoring metrics. Section \ref{ml} shows the loss functions of machine learning and how they fit nonlinear payoffs in practical applications.

The appendix shows the mathematical derivations and exact distribution of the various payoffs, along with an exact explicit functions for the Brier score helpful for other applications such as significance testing and sample sufficiency (new to the literature). 

\begin{figure}[h!]
	\includegraphics[width=\linewidth]{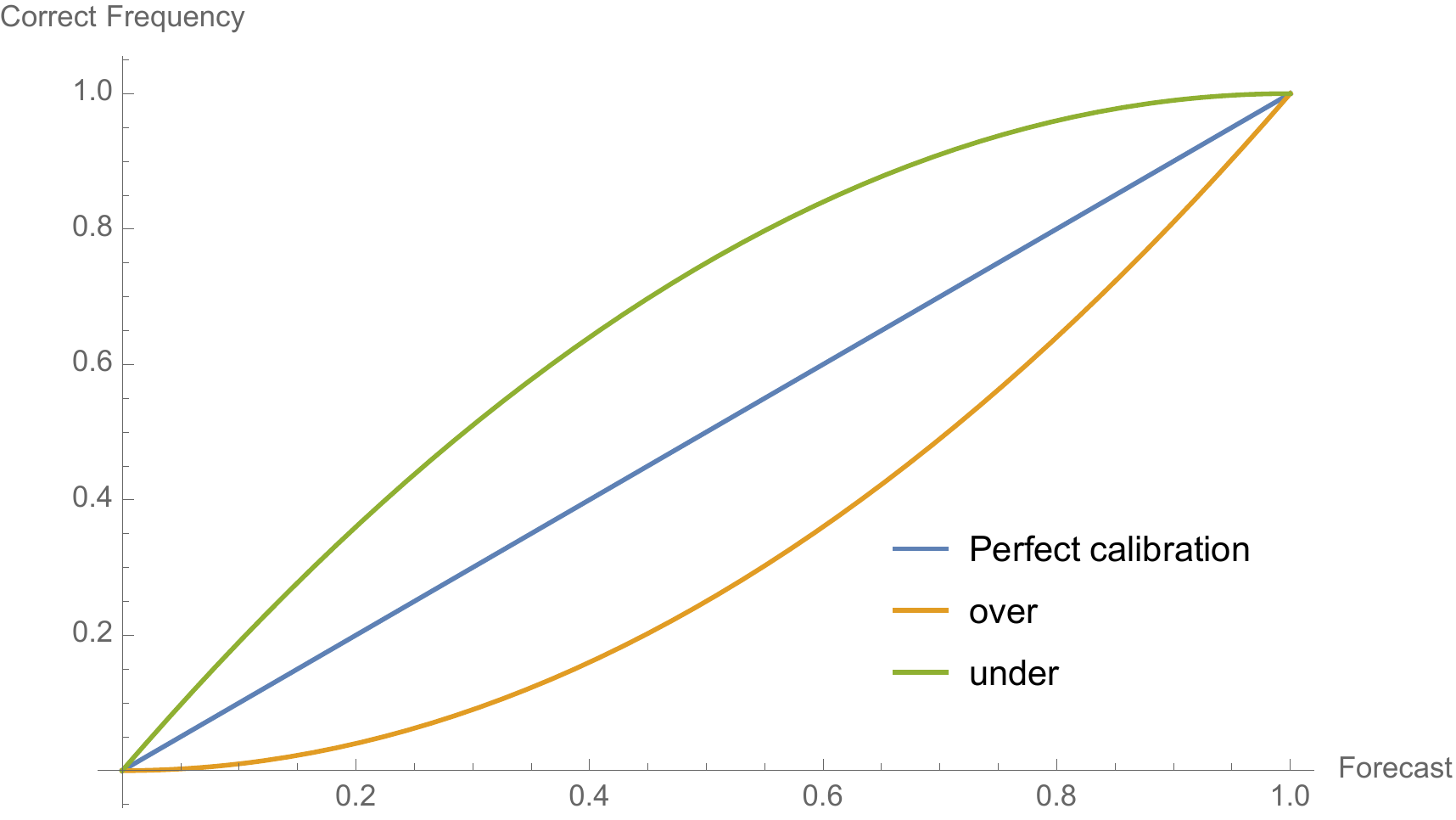}
	\caption{Probabilistic calibration (simplified) as seen in the psychology literature. The $x$ axis shows the estimated probability produced by the forecaster, the $y$ axis the actual realizations, so if a weather forecaster predicts $30\%$ chance of rain, and rain occurs $30\%$ of the time, they are deemed "calibrated". We hold that calibration in frequency (probability) space is an academic exercise (in the bad sense of the word) that mistracks real life outcomes outside narrow binary bets. It is particularly fallacious under fat tails.}\label{calibration}
\bigskip
	\includegraphics[width=\linewidth]{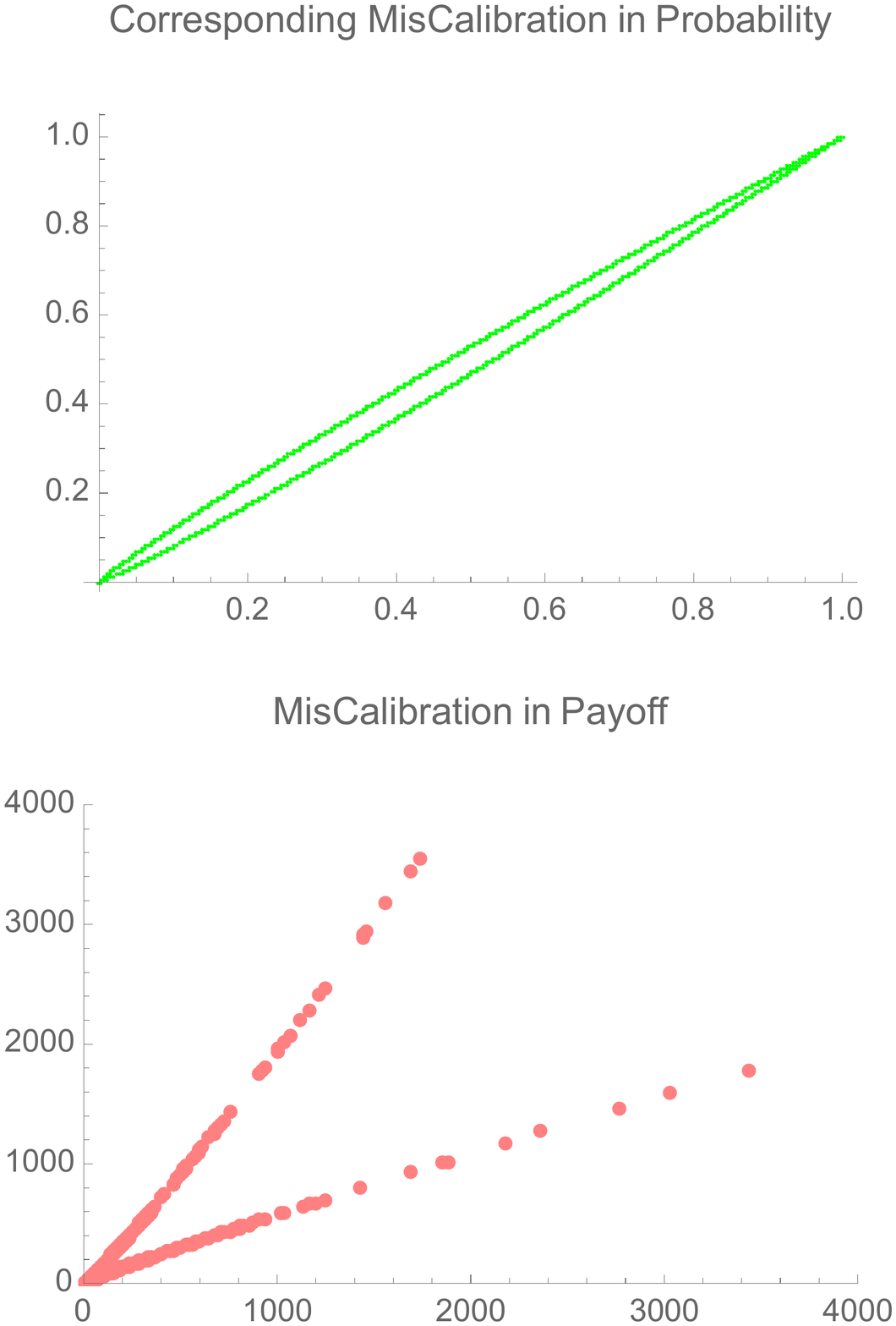}
	\caption{How miscalibration in probability corresponds to miscalibration in payoff under power laws. The distribution under consideration is Pareto with tail index $\alpha=1.15$}\label{miscalibration}
\end{figure}

\begin{figure}[h!]
	\includegraphics[width=\linewidth]{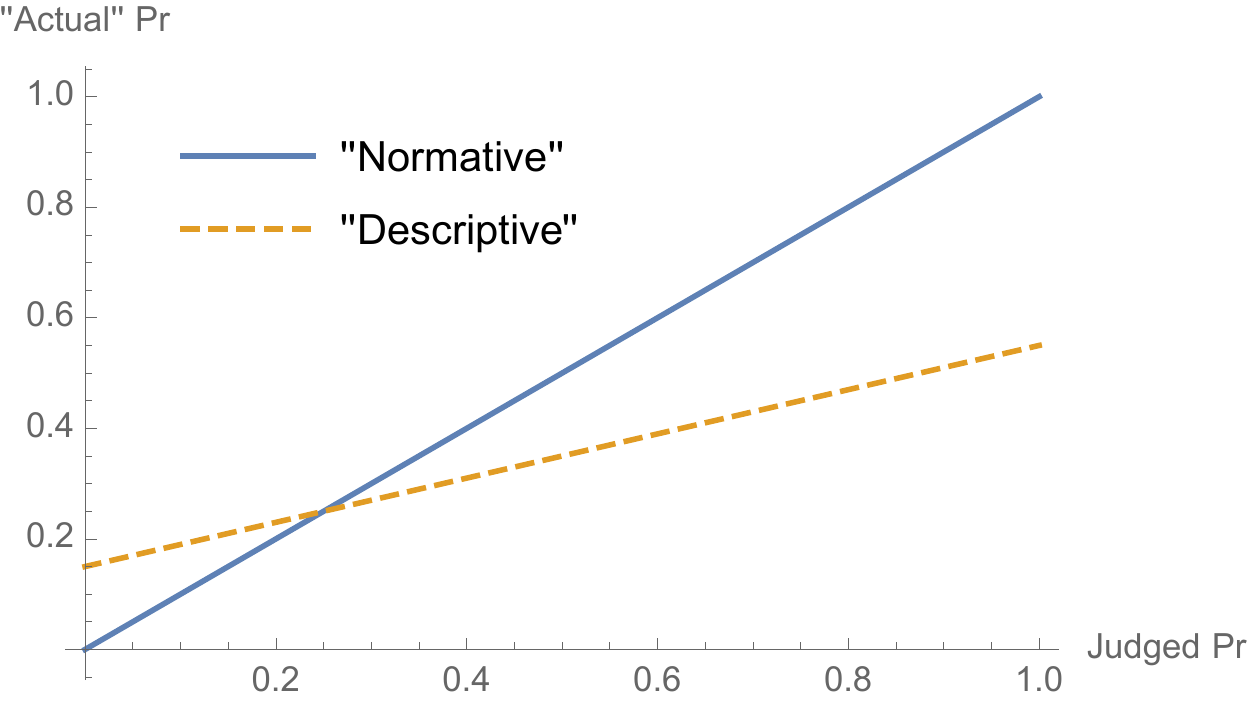}\label{baron}
	\caption{"Typical patterns" as stated and described in Baron's textbook \cite{baron2008thinking}, a representative claim in psychology of decision making that people overestimate small probability events. 
	We note that to the left, in the estimation part, 1) events such as floods, tornados, botulism, mostly patently fat tailed variables, matters of severe consequences that agents might have incorporated in the probability, 2) these probabilities are subjected to estimation error that, when endogenized, increase the true probability of remote events.}
	\end{figure}

\section{Continuous vs. Discrete Payoffs: Definitions and Comments}\label{contdiscrete}
\begin{example}["One does not eat beliefs and (binary) forecasts"]
In the first volume of the \textit{Incerto} (\textit{ Fooled by Randomness}, 2001 \cite{taleb2016incerto}), the narrator, a trader, is asked by the manager "do you predict that the market is going up or down?" "Up", he answered, with confidence. Then the boss got angry when, looking at the firm's exposures, he discovered that the narrator was short the market, i.e., would benefit from the market going down. 

The trader had difficulties conveying the idea that there was no contradiction, as someone could hold the (binary) belief that the market had a higher probability of going up than down, but that, should it go down, there is a very small probability that it could go down considerably, hence a short position had a positive expected return and the rational response was to engage in a short exposure. "You do not eat forecasts, but P/L"  (or "one does not monetize forecasts") goes the saying among traders.
\end{example}
If exposures and beliefs do not go in the same direction, it is because beliefs are verbalistic reductions that contract a higher dimensional object into a single dimension. To express the manager's error in terms of decision-making research, there can be a conflation in something as elementary as the notion of a binary event (related to the zeroth moment) or the \textit{probability} of an event and \textit{expected payoff} from it (related to the first moment and, when nonlinear, to all higher moments) as the payoff functions of the two can be similar in some circumstances and different in others. 

\begin{commentary}
	In short, probabilistic calibration requires estimations of the zeroth moment while the real world requires all moments (outside of gambling bets or artificial environments such as psychological experiments where payoffs are necessarily truncated), and it is a central property of fat tails that higher moments are explosive (even "infinite") and count more and more.
\end{commentary}



\subsection{Away from the Verbalistic}
While the trader story is mathematically trivial (though the mistake is committed a bit  too often), more serious gaps are present in decision making and risk management, particularly when the payoff function is more complicated, or nonlinear (and related to higher moments). So once we map the contracts or exposures mathematically, rather than focus on words and verbal descriptions, some serious distributional issues arise.
\begin{definition}[ Event]
	A (real-valued) random variable $X\colon \Omega \to \mathbb{R}$ defined on the probability space $(\Omega, \mathcal{F}, P)$ is a function $X(\omega)$ of the outcome $\omega \in \Omega$. An event is a measurable subset (countable or not) of $\Omega$, measurable meaning that it can be defined through the value(s) of one of several random variable(s). 
	\end{definition}

\begin{definition}[Binary forecast/payoff]
A
binary forecast (belief, or payoff) is a random variable  taking two  values $$X: \Omega\rightarrow \{X_1,X_2\},$$ with realizations $X_1,X_2\in \mathbb{R}$.

\end{definition}
In other words, it lives in the binary set (say $\{0,1\}$, $\{-1,1\}$, etc.), i.e., the specified event will or will not take place and, if there is a payoff, such payoff will be mapped into two finite numbers (a fixed sum if the event happened, another one if it didn't). Unless otherwise specified, in this discussion we default to the $\{0,1\}$ set.

Example of situations in the real world where the payoff is binary:
        \begin{itemize}
\item Casino gambling, lotteries	, coin flips, "ludic" environments, or binary options paying a fixed sum if, say, the stock market falls below a certain point and nothing otherwise --deemed a form of gambling\footnote{Retail binary options are typically used for gambling and have been banned in many jurisdictions, such as, for instance, by the European Securities and Markets Authority (ESMA), www.esma.europa.eu, as well as the United States where it is considered another form of internet gambling, triggering a complaint by a collection of decision scientists, see Arrow et al. \cite{arrow2008promise}. We consider such banning as justified since bets have practically no economic value, compared to financial markets that are widely open to the public, where natural exposures can be properly offset.}.
\item Elections where the outcome is binary (e.g., referenda, U.S. Presidential Elections), though not the economic effect of the result of the election.\footnote{Note the absence of spontaneously forming gambling markets with binary payoffs for continuous variables. The exception might have been binary options but these did not remain in fashion for very long, from the experiences of the author, for a period between 1993 and 1998, largely motivated by tax gimmicks.}
\item Medical prognoses for a single patient entailing survival or cure over a specified duration, though not the duration itself as variable, or disease-specific survival expressed in time, or conditional life expectancy. Also exclude anything related to epidemiology.
\item Whether a given person who has an online profile will buy or not a unit or more of a specific product at a given time (not the quantity or units).
        \end{itemize}
\begin{commentary}[A binary belief is equivalent to a payoff]
A binary "belief" should map to an economic payoff (under some scaling or normalization necessarily to constitute a probability), an insight owed to De Finetti \cite{de1972probability}  who held that a "belief" and a "prediction" (when they are concerned with two distinct outcomes) map into the equivalent of the expectation of a binary random variable and bets with a payoff in $\{0,1\}$. An "opinion" becomes a choice price for a gamble, and one at which one is equally willing to buy or sell. Inconsistent opinions therefore would lead to a violation of arbitrage rules, such as the "Dutch book", where a combination of mispriced bets can guarantee a future loss.
\end{commentary}

\begin{definition}[Real world open continuous payoff]
$$X:\Omega \to [a, \infty) \lor (-\infty,b] \lor (-\infty, \infty) $$
A continuous payoff "lives" in an interval, not a finite set. It  corresponds to an unbounded random variable  either doubly unbounded or semi-bounded, with the bound on one side (one tailed variable).
\end{definition}
\subsubsection*{Caveat} We are limiting for the purposes of our study the consideration to binary vs. continuous and open-ended (i.e., no compact support). Many discrete payoffs are subsumed into the continuous class using standard arguments of approximation. We are also omitting triplets, that is, payoffs in, say $\{-1,0,3\}$, as these obey the properties of binaries (and can be constructed using a sum of binaries). Further, many variable with a floor and a remote ceiling (hence, formally with compact support), such as the number of victims or a catastrophe, are analytically and practically treated as if they were open-ended \cite{cirillo2016statistical}.

Example of situations in the real world where the payoff is continuous:
        \begin{itemize}
\item Wars casualties, calamities due to earthquake, medical bills, etc.
\item Magnitude of a market crash, severity of a recession, rate of inflation
\item Income from a strategy
\item Sales and profitability of a new product \item In general, anything covered by an insurance contract
         \end{itemize}

\begin{figure}[h!]
\includegraphics[width=\columnwidth]{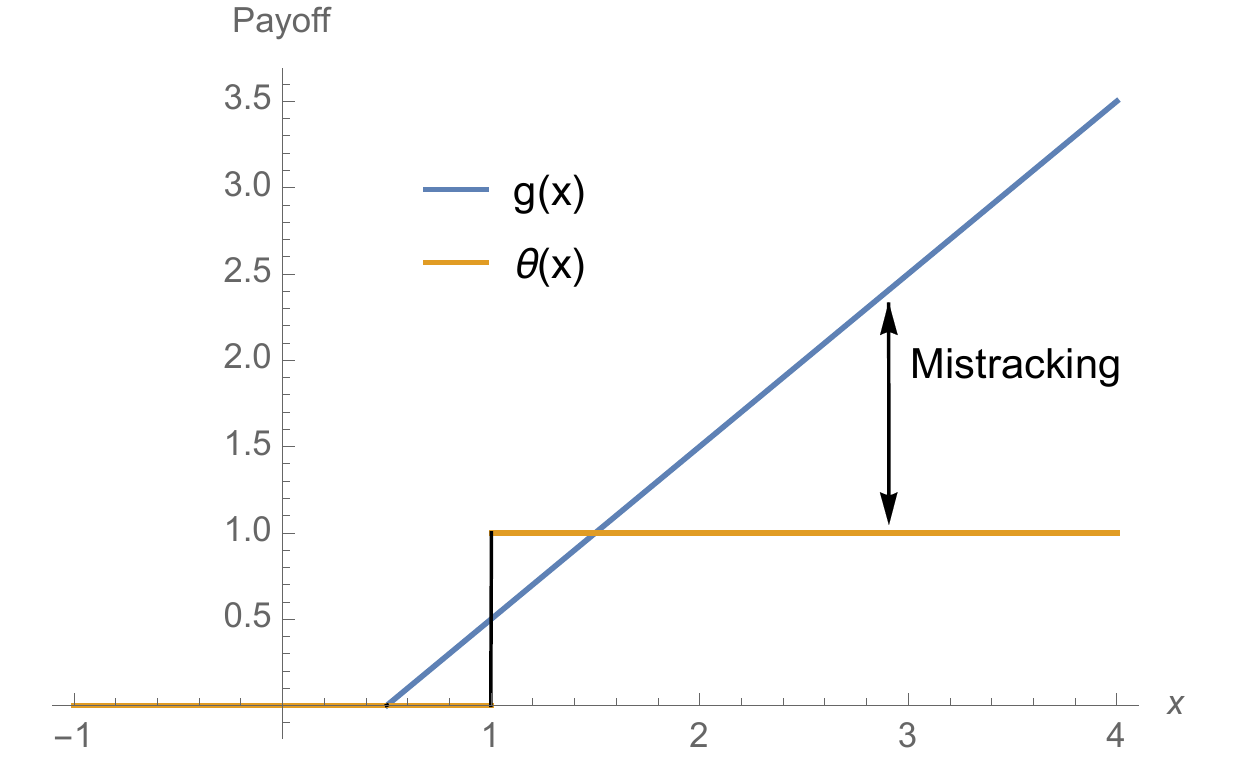}
\caption{ Comparing the payoff of a binary bet (The Heaviside $\theta(.)$) to a continuous open-ended exposure $g(x)$. Visibly there is no way to match the (mathematical) derivatives for any form of hedging. }\label{mistrackingbinary}
\end{figure}

Most natural and socio-economic variables are continuous and their statistical distribution does not have a compact support in the sense that we do not have a handle of an exact upper bound.

\begin{example}
Predictive analytics in binary space $\{0,1\}$ can be successful in forecasting if, from his online activity, online consumer Iannis Papadopoulos will purchase a certain item, say a wedding ring, based solely on computation of the probability. But the probability of the "success" for a potential new product might be --as with the trader's story-- misleading. Given that company sales are typically fat tailed, a very low probability of success might still be satisfactory to make a decision. Consider venture capital or option trading --an out of the money option can often be attractive yet may have less than 1 in 1000 probability of ever paying off.

More significantly, the tracking error for probability guesses will not map to that of the performance. 
\end{example}

This difference is well known by option traders as there are financial derivative contracts called "binaries" that pay in the binary set $\{0,1\}$ (say if the underlying asset $S$, say, exceeds a strike price $K$), while others called "vanilla" that pay in $[0,\infty)$, i.e. $\max(S-K,0)$ (or, worse, in $(-\infty, 0)$ for the seller can now be exposed to bankruptcy owing to the unbounded exposure). The considerable mathematical and economic difference between the two has been discussed and is the subject of \textit{Dynamic Hedging: Managing Vanilla and Exotic Options} \cite{taleb1997dynamic}. Given that the former are bets paying a fixed amount and the latter have full payoff, one cannot be properly replicated (or hedged) using another, especially under fat tails and parametric uncertainty --meaning performance in one does not translate to performance into the other.  While this knowledge is well known in mathematical finance it doesn't seem to have been passed on to the decision-theory literature.
\begin{commentary}[Derivatives theory]
Our approach here is inspired from derivatives (or option) theory and practice where there are different types of derivative contracts, 1) those with binary payoffs (that pay a fixed sum if an event happens) and 2) "vanilla" ones (standard options with continuous payoffs). It is practically impossible to hedge one with another \cite{taleb1997dynamic}. Furthermore a binary  bet with a strike price $K$ and a call option with same strike $K$, with $K$ in the tails of the distribution, almost always have their valuations react in opposite ways when one increases the kurtosis of the distribution, (while preserving the first three moments) or, in an example further down in the lognormal environment, when one increases uncertainty via the scale of the distribution.
\end{commentary}
\begin{commentary}[Term sheets]
Note that, thanks to "term sheets" that are necessary both legally and mathematically, financial derivatives practice provides precise legalistic mapping of payoffs in a way to make their mathematical, statistical, and economic differences salient.
\end{commentary}

There has been a  tension between prediction markets and real financial markets. As we can show here, prediction markets may be useful for gamblers, but they cannot hedge economic exposures.

The mathematics of the difference and the impossibility of hedging can be shown in the following. Let $X$ be a random variable in $\mathbbm{R}$, we have the payoff of the bet or the prediction $\theta_K:\mathbbm{R}\rightarrow \{0,1\}$,
\begin{equation}
\theta_K(x)=\left\{\begin{array}{cc}
 	1, &  x\geq K \\
 	 0 & \text{otherwise},\\
 \end{array}	 \right.	
\end{equation}
 and $g: \mathbbm{R} \rightarrow \mathbbm{R}$ that of the natural exposure. Since $\frac{\partial }{\partial x}\theta_K(x)$ is a Dirac delta function at $K$, $\delta(K)$  and $\frac{\partial)}{\partial x} g_k(x)$ is at least once differentiable for $x\geq K$ (or constant in case the exposure is globally linear or, like an option, piecewise linear above $K$), matching derivatives for the purposes of offsetting variations is not a possible strategy.\footnote{To replicate an open-ended continuous payoff with binaries, one  needs an infinite series of bets, which cancels the entire idea of a prediction market by transforming it into a financial market. Distributions with compact support always have finite moments, not the case of those on the real line.} The point is illustrated in Fig \ref{mistrackingbinary}.

\section{ There is no defined "collapse",  "disaster", or "success" under fat tails} \label{nocollapse}

The fact that an "event" has some uncertainty around its magnitude carries some mathematical consequences. Some verbalistic papers  still commit in 2019 the fallacy of binarizing an event in $[0,\infty)$: A recent paper on calibration of beliefs, \cite{dana2019markets} says "...if a person claims that the United States is on the verge of an economic collapse or that a climate disaster is imminent..." An economic "collapse" or a climate "disaster" must not be expressed as an event in $\{0,1\}$ when in the real world it can take many values. For that, a characteristic scale is required. In fact under fat tails, there is no "typical" collapse or disaster, owing to the absence of characteristic scale, hence verbal binary predictions or beliefs cannot be used as gauges. 

We present the difference between thin tailed and fat tailed domains (illustrated for the intuition of the differences in Fig \ref{threeplots}) as follows.
 \begin{definition}[Characteristic scale]
 Let $X$ be a random variable that lives in either $(0,\infty)$ or $(-\infty, \infty)$ and $\mathbbm{E}$ the expectation operator under "real world" (physical) distribution. By classical results \cite{embrechts1997modelling}:
 \begin{equation}
 	\lim_{K \to \infty} \frac{1}{K} \mathbbm{E}(X|_{X>K})= \lambda,\label{chscale}
 \end{equation}
 \begin{itemize}
\item If 	$\lambda=1$ , $X$ is said to be in the thin tailed class $\mathcal{D}_1$ and has a characteristic scale
\item If 	$\lambda>1$ , $X$ is said to be in the fat tailed regular variation class $\mathcal{D}_2$ and has no characteristic scale
\item If   $$\lim_{K \to \infty}  \mathbbm{E}(X|_{X>K})-K= \mu$$ where $\mu >0$,  then $X$ is in the borderline exponential class
 \end{itemize}

\end{definition}

 The point can be made clear as follows. One cannot have a binary contract that adequately hedges someone against a "collapse", given that one cannot know in advance the size of the collapse or how much the face value or such contract needs to be. On the other hand, an insurance contract or option with continuous payoff would provide a satisfactory hedge. Another way to view it: reducing these events to verbalistic "collapse", "disaster" is equivalent to a health insurance payout of a lump sum if one is "very ill" --regardless of the nature and gravity of the illness -- and $0$ otherwise.

And it is highly flawed to separate payoff and probability in the integral of expected payoff.\footnote{ Practically all economic and informational variables have been shown since the 1960s to belong to the $\mathcal{D}_2$ class, or at least the intermediate subexponential class (which includes the lognormal), \cite{mandelbrot1960pareto,mandelbrot1963stable,mandelbrot1997new, gabaix2008power,taleb2016incerto}, along with social variables such as size of cities, words in languages, connections in networks, size of firms, incomes for firms, macroeconomic data, monetary data, victims from interstate conflicts and civil wars\cite{richardson1941frequency, cirillo2016statistical}, operational risk, damage from earthquakes, tsunamis, hurricanes and other natural calamities, income inequality \cite{champernowne1953model}, etc. Which leaves us with the more rational question: where are Gaussian variables? These appear to be at best one order of magnitude fewer in decisions entailing formal predictions.} Some experiments of the type shown in Fig. \ref{baron} ask agents what is their estimates of deaths from botulism or some such disease: agents are blamed for misunderstanding the probability. This is rather a problem with the experiment: people do not necessarily separate probabilities from payoffs.

\section{Spurious overestimation of tail probability in the psychology literature}\label{spurious}

	\begin{definition}[Substitution of integral]
	  Let $K \in \mathbbm{R}^+$ be a threshold, $f(.)$ a density function and $p_K \in [0,1]$ the probability of exceeding it, and $g(x)$ an impact function. Let $I_1$ be the expected payoff above $K$:
$$I_1=\int_K^{\infty } g(x) f(x) \, \mathrm{d}x,$$ 
and Let $I_2$ be the impact at $K$ multiplied by the probability of exceeding $K$:
 $$ I_2=g(K) \int_K^{\infty } f (x) \, dx=g(K) p_K.$$ 
The substitution comes from conflating $I_1$ and $I_2$, which becomes an identity if and only if $g(.)$ is constant above $K$ (say $g(x)=\theta_K(x)$, the Heaviside theta function).  For $g(.)$ a variable function with positive first derivative, $I_1$ can be close to $I_2$ only under thin-tailed distributions, not under the fat tailed ones.

\end{definition}
For the discussions and  examples in this section assume $g(x) = x$ as we will consider the more advanced nonlinear case in Section \ref{ml}. 

\begin{theorem}[Convergence of $\frac{I_1}{I_2}$]
If $X$ is in the thin tailed class $\mathcal{D}_1$ as described in \ref{chscale},
\begin{equation}
\lim_{K\to \infty} \frac{I_1}{I_2}= 1		
\end{equation}
If $X$ is in the regular variation class $\mathcal{D}_2$,
\begin{equation}
\lim_{K\to \infty} \frac{I_1}{I_2}=\lambda > 1.		
\end{equation}

\end{theorem}

\begin{proof}
From Eq. \ref{chscale}.
Further comments:
\subsection{Thin tails}  By our very definition of a thin tailed distribution (more generally any distribution outside the subexponential class, indexed by $(g)$), where $f^{(g)}(.)$ is the PDF:
\begin{equation}
	\underset{K\to \infty }{\text{lim}}\frac{\int_{K}^{\infty } x f^{(g)} (x) \, dx}{K \int_{K}^{\infty } f^{(g)} (x) \, dx }=\frac{I_1}{I_2}=1.
\end{equation}
Special case of a Gaussian: Let $g(.)$ be the PDF of predominantly used Gaussian distribution (centered and normalized),
\begin{equation}
	\int_K^{\infty } x g (x) \, dx=\frac{e^{-\frac{K^2}{2}}}{\sqrt{2 \pi }}
\end{equation}
and $K_p=\frac{1}{2} \text{erfc}\left(\frac{K}{\sqrt{2}}\right)$, where erfc is the complementary error function, and $K_p$ is the threshold corresponding to the probability $p$. 

We note that $K_p \frac{I_1}{I_2}$ corresponds to the inverse Mills ratio used in insurance.
\subsection{Fat tails} For all distributions in the regular variation class, defined by their tail survival function: for $K$ large,  $$\mathbb{P} (X>K) \approx L K^{-\alpha },\;\alpha >1,$$
 where $L>0$ and $f^{(p)}$ is the PDF of a member of that class:
\begin{equation}
	\underset{K_p\to \infty }{\text{lim}}\frac{\int_K^{\infty } x f^{(p)} (x) \, dx}{K \int_{K_p}^{\infty } f^{(p)} (x) \, dx}=\frac{\alpha }{\alpha -1}>1
\end{equation}

\end{proof}
\subsection{Conflations}

\subsubsection{Conflation of $I_1$ and $I_2$} 
In numerous experiments, which include the prospect theory paper by Kahneman and Tversky (1978) \cite{kahneman1979prospect}, it has been repeatedly established that agents overestimate small probabilities in experiments where the odds are shown to them, and when the outcome corresponds to a single payoff. The well known Kahneman-Tversky result proved robust, but interpretations make erroneous claims from it. Practically all the subsequent literature, relies on $I_2$ and conflates it with $I_1$, what this author has called \textit{the ludic fallacy} in \textit{The Black Swan} \cite{taleb2016incerto}, as games are necessarily truncating a dimension from reality.  The psychological results might be robust, in the sense that they replicate when repeated in the exact similar conditions, but all the claims outside these conditions and extensions to real risks will be an exceedingly dubious generalization --given that our exposures in the real world  rarely map to $I_1$. Furthermore, one can overestimate the probability yet underestimate the expected payoff. 
\subsubsection{Stickiness of the conflation} The misinterpretation is still made four decades  after Kahneman-Tversky (1979). In a review of behavioral economics, with emphasis on miscaculation of probability, Barberis (2003) \cite{barberis2013psychology} treats $I_1=I_2$. And Arrow et al. \cite{arrow2008promise}, a long list of decision scientists pleading for deregulation of the betting markets also misrepresented the fitness of these binary forecasts to the real world (particularly in the presence of real financial markets). 

\subsubsection{The Value at Risk Problem} Another stringent --and  dangerous --example is the "default VaR" (Value at risk, i.e. the minimumm one is expected to lose within, say, 1\% probability, over a certain period, that is the quantile for 1\%, see \cite{artzner1999coherent}) which is explicitly given as $I_2$, i.e. default probability $x (1 - \text{expected recovery rate})$. This quantity can be quite different from the actual loss expectation \textit{in case} of default. The Conditional Value at Risk, CVar, measures the expected total losses conditional on default, which is altogether another entity, particularly under fat tails (and, of course, considerably higher). Further, financial regulators presents erroneous approximations of CVaR, and the approximation is the risk-management flaw that may have caused the crisis of 2008 \cite{taleb2012prevent}.\footnote{The mathematical expression of the Value at Risk, VaR, for a random variable $X$ with distribution function $F$ and threshold $\alpha \in [0,1]$ $$\operatorname{VaR}_\alpha(X)=-\inf\big\{x\in\mathbb{R}:F_X(x)>\alpha\big\} 
,$$ and the corresponding CVar
$$
 \operatorname {ES}_\alpha(X) = \mathbb{E}\left(-X\mid_{X \leq -\operatorname{VaR}_\alpha(X)}\right) 
$$}
The fallacious argument is that they compute the recovery rate as the expected value of collateral, without conditioning by the default event. The expected value of the collateral conditionally to a default is often far less then its unconditional expectation. In 2007, after a massive series of foreclosures, the value of most collaterals dropped to about 1/3 of the expected value!

Furthermore,  VaR is not a coherent risk measure, owing to its violating the subadditivity property \cite{delbaen1994general}: for random variables (or assets) $X$ and $Y$ and a threshold $\alpha$, a measure must meet the inequality $\mu_\alpha(X+Y)\leq \mu_\alpha(X)+\mu_\alpha(Y)$. On the other hand, CVaR satisfies the coherence bound; likewise $I_2$, because its construction is similar to VaR, violates the inequality.

\subsubsection{Misunderstanding of Hayek's knowledge arguments}  "Hayekian" arguments for the consolidation of beliefs via prices does not lead to prediction markets as discussed in such pieces as \cite{bragues2009prediction}, or Sunstein's \cite{sunstein2006deliberating}: prices exist in financial and commercial markets; prices are not binary bets. For Hayek \cite{hayek1945use} consolidation of knowledge is done via prices and \textit{arbitrageurs} (his words)--and arbitrageurs trade products, services, and financial securities, not binary bets.

\begin{definition}[Corrected probability in binarized experiments]
	Let $p^*$ be the equivalent probability to make $I_1=I_2$ and eliminate the effect of the error, so $$p^*=\{p: I_1=I_2 \}$$
\end{definition}
Now let' s solve for $K_p$ "in the tails", working with a probability $p$. For the Gaussian, $K_p=\sqrt{2} \text{erfc}^{-1}(2 p)$; for the Paretan tailed distribution, $K_p=p^{-1/\alpha }$.

Hence, for a Paretan distribution, the ratio of real continuous probability to the binary one $$\frac{p^*}{p}=\frac{\alpha}{\alpha-1}$$
which can allow in absurd cases $p^*$ to exceed 1 when the distribution is grossly misspecified.	

Tables \ref{Gaussianpseudooverestimation} and \ref{Paretanpseudooverestimation} show, for a probability level $p$, the corresponding tail level $K_p$, such as $$K_p=\left\{\text{inf } K: \mathbbm{P}(X>K)> p \right\},$$ and the corresponding adjusted probability $p^*$ that de-binarize the event \footnote{The analysis is invariant to whether we use the right or left tail .By convention, finance uses negative value for losses, whereas other areas of risk management express the negative of the random variance, hence focus on the right tail.}\footnote{$K_p$ is equivalent to the Value at Risk $VaR_p$ in finance, where $p$ is the probability of loss.}-- probabilities here need to be in the bottom half, i.e., $p<.5$. Note that we are operating under the mild case of known probability distributions, as it gets worse under parametric uncertainty.\footnote{Note the van der Wijk's law, see Cirillo \cite{cirillo2013your}: $\frac{I_1}{I_2}$ is related to what is called in finance the expected shortfall for $K_p$.}
 \begin{figure}
\includegraphics[width=1.1\linewidth]{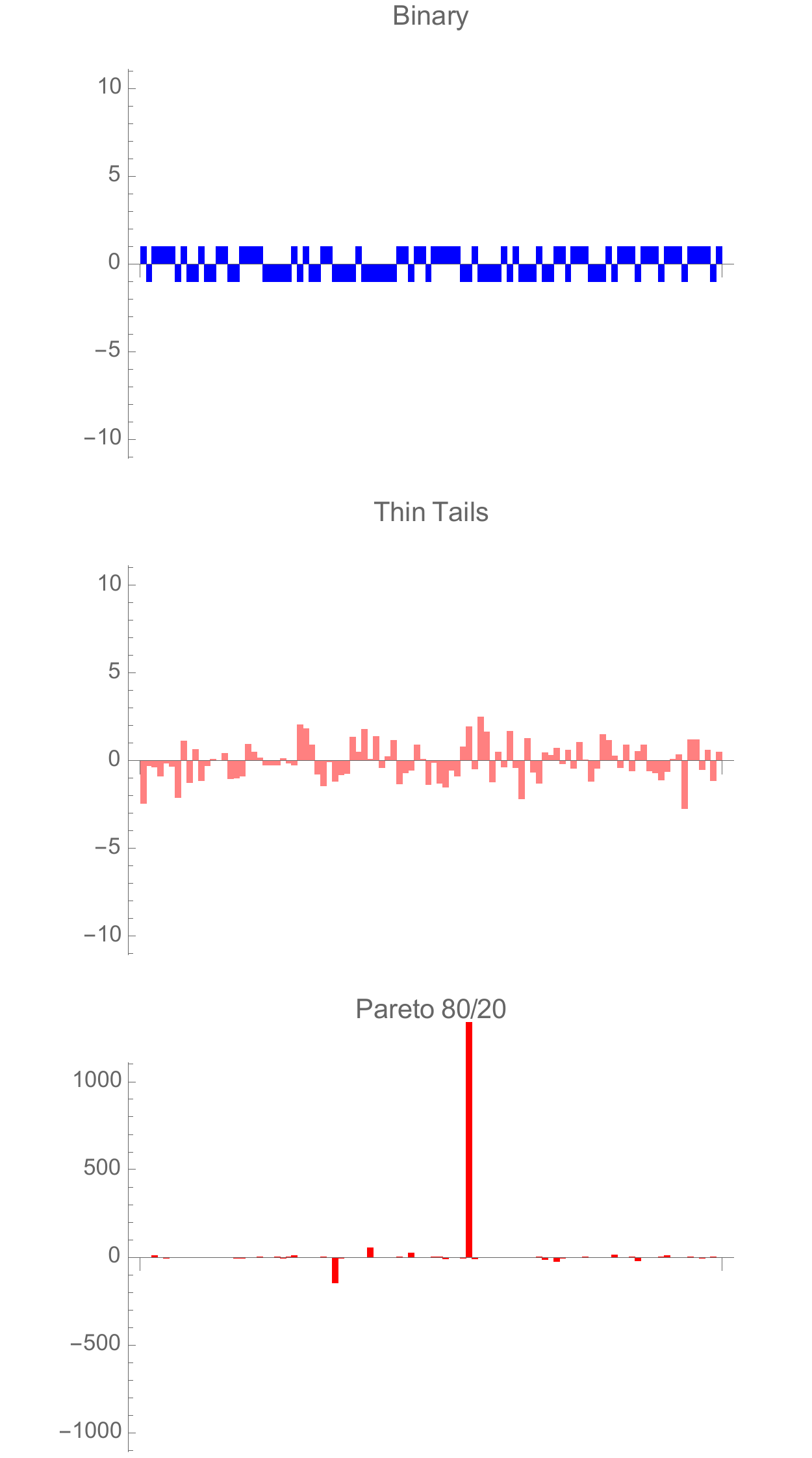}\label{threeplots}
\caption{Comparing the three payoff time series under two distributions --the binary has the same profile regardless of whether the distribution is thin or fat tailed. The first two subfigures are to scale, the third (representing the Pareto 80/20 with $\alpha= 1.16$ requires multiplying the scale by two orders of magnitude.}
\label{threeplots}
\end{figure}

The most commonly known distribution among the public, the "Pareto 80/20" (based on Pareto discovering that 20 percent of the people in Italy owned 80 percent of the land), maps to a tail index $\alpha=1.16$, so the adjusted probability is $>7$ times the naive one. 

\begin{table}
\caption{Gaussian pseudo-overestimation}	\label{Gaussianpseudooverestimation}
\begin{tabular}{l|lllll}
 p & $K_p$ & $\int _{K_p}^{\infty }x f(x)dx$ & $K_p\int _{K_p}^{\infty }f(x)dx$ & $p^* $&
 $\frac{p^*}{p} $\\
    \hline
$ \frac{1}{10}$ & $1.28$ & $1.75\times 10^{-1}$ & $1.28\times 10^{-1}$ & $1.36\times 10^{-1}$ & $1.36$
   \\
 $\frac{1}{100}$ & $2.32$ & $2.66\times 10^{-2}$ & $2.32\times 10^{-2}$ & $1.14\times 10^{-2}$ & 1.14
   \\
$ \frac{1}{1000} $& 3.09 & $3.36\times 10^{-3}$ & $3.09\times 10^{-3}$ & $1.08\times 10^{-3} $&
   1.08 \\
$ \frac{1}{10000}$ & 3.71 & $3.95\times 10^{-4}$ & $3.71\times 10^{-4}$ & $1.06\times 10^{-4}$ &
   1.06 \\
\end{tabular}
\end{table}

\begin{table}
\caption{Paretan pseudo-overestimation}	\label{Paretanpseudooverestimation}
\begin{tabular}{l|lllll}
 p & $K_p$ & $\int _{K_p}^{\infty }x f(x)dx$ & $K_p\int _{K_p}^{\infty }f(x)dx$ & $p^* $&
 $\frac{p^*}{p} $\\
 \hline
 $\frac{1}{10}$ & 8.1 & 8.92 & 0.811 & 1.1 (sic) & 11. \\
 $\frac{1}{100}$ & 65.7 & 7.23 & 0.65 & 0.11 & 11. \\
 $\frac{1}{1000}$ & 533 & 5.87 & 0.53 & 0.011 & 11. \\
 $\frac{1}{10000} $& 4328 & 4.76 & 0.43 & 0.0011 & 11. \\	
\end{tabular}
\end{table}	

\subsubsection{Example of probability and expected payoff reacting in opposite direction under increase in uncertainty}
An example showing how, under a skewed distribution, the binary and the expectation reacting in opposite directions is as follows. Consider the risk-neutral lognormal distribution $\mathcal{L}(X_0-\frac{1}{\sigma^2},\sigma)$ with pdf $f_L(.)$, mean $X_0$ and variance $\left(e^{\sigma ^2}-1\right) X_0^2$. We can increase its uncertainty with the parameter $\sigma$. We have the expectation of a contract above $X_0$, $\mathbbm{E}_{>X_0}$:
 $$\mathbbm{E}_{>X_0}=\int_{X_0}^\infty x f_L(x) \;\mathrm{d}x= \frac{1}{2} X_0 \left(1+\text{erf}\left(\frac{\sigma }{2 \sqrt{2}}\right)\right)$$ 
and the probability of exceeding $X_0$, 
$$\mathbbm{P}(X>X_0)=\frac{1}{2} \left(1-\text{erf}\left(\frac{\sigma }{2 \sqrt{2}}\right)\right),$$ 
where erf is the  error function. As $\sigma$ rises $\text{erf}\left(\frac{\sigma }{2 \sqrt{2}}\right) \to 1$, with $\mathbbm{E}_{>X_0} \to X_0$ and $\mathbbm{P}(X>X_0) \to 0$. This example is well known by option traders (see \textit{Dynamic Hedging} \cite{taleb1997dynamic}) as the binary option struck at $X_0$ goes to 0 while the standard call of the same strike rises considerably to reach the level of the asset --regardless of strike. This is typically the case with venture capital: the riskier the project, the less likely it is to succeed but the more rewarding in case of success. So, the expectation can go to $+\infty$ while to probability of success goes to $0$.

\subsection{Distributional Uncertainty}
We can gauge the effect of distributional uncertainty by examining the second order effect with respect to a parameter, which shows if errors induce either biases or an acceleration of divergence. A stable model needs to be necessarily linear to errors, see \cite{taleb2013mathematical}. \footnote{Furthermore, distributional uncertainty is in by itself a generator of fat tails: heteroskedasticity or variability of the variance produces higher kurtosis in the resulting distribution \cite{taleb1997dynamic},\cite{gatheral2006volatility}, \cite{DemDerKam99}.}

\begin{remark}[Distributional uncertainty]
	Owing to Jensen's inequality, the discrepancy $\left(I_1-I_2\right)$ increases under parameter uncertainty, expressed in higher kurtosis, via stochasticity of $\sigma$ the scale of the thin-tailed distribution, or that of $\alpha$ the tail index of the Paretan one.
\end{remark}

\begin{proof}
First, the Gaussian world. We consider the effect of $I_1-I_2= \int_K^\infty x f^{(g)}(x)-\int_K^\infty f^{(g)}(x) $ under stochastic volatility, i.e. the parameter from increase of volatility. Let $\sigma$ be the scale of the Gaussian, with $K$ constant:
\begin{dmath}
\frac{\partial ^2 (\int_K^\infty x f^{(g)}(x) dx)}{\partial \sigma ^2}-\frac{\partial ^2 (\int_K^\infty  f^{(g)}(x) dx)}{\partial \sigma ^2} =\frac{e^{-\frac{K^2}{2 \sigma ^2}} \left((K-1) K^3-(K-2) K \sigma ^2\right)}{\sqrt{2 \pi } \sigma ^5},
\end{dmath}
which is positive for all values of $K >0$ (given that $K^4-K^3-K^2+2 K>0$ for $K$ positive).

Second, consider the sensitivity of the ratio $\frac{I_1}{I_2}$ to parameter uncertainty for $\alpha$ in the Paretan case (for which  we can get a streamlined expression compared to the difference). For $\alpha>1$ (the condition for a finite mean):
\begin{dmath}
\frac{\partial ^2 \left(\int_K^\infty x f^{(p)}(x) dx/\int_K^\infty  f^{(p)}(x) dx\right)}{\partial \alpha ^2}=
\frac{2 K}{(\alpha -1)^3}
\end{dmath}
which is positive and increases markedly at lower values of $\alpha$, meaning the fatter the tails, the worse the uncertainty about the expected payoff and the larger the difference between $I_1$ and $I_2$.

\end{proof}

\section{Calibration and Miscalibration}\label{miscalibr}
The psychology literature also examines the "calibration" of probabilistic assessment --an evaluation of how close someone providing odds of events turns out to be on average (under some operation of the law of large number deemed satisfactory) \cite{lichtenstein1977calibration}, \cite{keren1991calibration}, see Figures \ref{calibration} , \ref{miscalibration}, and \ref{baron}. The methods, for the reasons we have shown here, are highly flawed except in narrow circumstances of purely binary payoffs (such as those entailing a "win/lose" outcome) --and generalizing from these payoffs is either not possible or produces misleading results. Accordingly, Fig. \ref{baron} makes little sense empirically.

At the core, calibration metrics such as the Brier score are always thin-tailed, when the variable under measurement is fat-tailed, which worsens the tractability.

To use again the saying "You do not eat forecasts", most businesses have severely skewed payoffs, so being calibrated in probability is meaningless.

 \begin{remark}[Distributional differences]
Binary forecasts and calibration metrics via the Brier score belong to the thin-tailed class.
 \end{remark}

\section{Scoring Metrics}\label{metrics}

\begin{table}\caption{Scoring Metrics for Performance Evaluation}\label{summarymetrics}

\begin{tabular}{|l|p{2. cm}|p{4.5cm}|}
\hline

\rowcolor{cyan!20}& & \\
\rowcolor{cyan!20}\textbf{Metric} & \textbf{Name} & \textbf{Fitness to reality} \\ 
\rowcolor{cyan!20}& & \\

\hline
$P^{(r)}(T)$& Cumulative P/L  & Adapted to real world distributions, particularly under a survival filter\\
\hline
$P^{(p)}(n)$&Tally of bets & Misrepresents the performance under fat tails, works only for binary bets and/or thin tailed domains.\\
\hline
$\lambda(n)$&Brier score & Misrepresents performance precision under fat tails, ignores higher moments.\\
\hline
$\lambda^{(M4)_{1,2}}_n $ &M4 first moment score (point estimate) & Represents precision, not exactly real world performance, but maps to real distribution of underlying variables.\\
\hline
$\lambda^{(M4)_3}_n $ &M4 dispersion Score & Represents precision in the assessment of confidence intervals.\\
\hline
	$\lambda^{(M5)}_n $ &Proposed M5 score  & Represents both precision and survival conditions by predicting extrema of time series.\\
	\hline
	$g(.)$ & Machine learning nonlinear payoff function (not a metric)& Expresses exposures without verbalism and reflects true economic or other P/L. Resembles financial derivatives term sheets.\\
\hline
\end{tabular}
	
\end{table}

This section, summarized in Table \ref{summarymetrics}, compares the probability distributions of the various metrics used to measure performance, either by explicit formulation or by linking it to a certain probability class. Clearly one may be mismeasuring performance if the random variable is in the wrong probability class. Different underlying distributions will require a different number of sample sizes owing to the differences in the way the law of numbers operates across distributions. A series of binary forecasts will converge very rapidly to a thin-tailed Gaussian even if the underlying distribution is fat-tailed, but an economic P/L tracking performance for someone with a real exposure will require a considerably larger sample size if, say, the underlying is Pareto distributed \cite{taleb2018much}.

We start by precise expressions for the four possible ones:
\begin{enumerate}
	\item Real world performance under conditions of survival, or, in other words, P/L or a quantitative cumulative score.
	\item A tally of bets, the naive sum of how often a person's binary prediction is correct
	\item De Finetti's Brier score $\lambda^{}(B)_{n}$
	\item The M4 point estimate score $\lambda^{M4}_{n}$ for $n$ observations used in the M4 competition, and its sequel M5.\footnote{https://www.m4.unic.ac.cy/wp-content/uploads/2018/03/M4-Competitors-Guide.pdf} \footnote{The M4 competition also entails a range estimate, evaluated according to techniques for the Mean Scaled Interval Score (MSIS), which reflects uncertainty about the variable and maps to the distribution of the mean absolute deviation.}
\end{enumerate}

\subsubsection{P/L in Payoff Space (under survival condition)}
The "P/L" is short for the natural profit and loss index, that is, a cumulative account of performance. Let $X_i$ be realizations of an unidimensional generic random variable $X$ with support in $\mathbb{R}$ and $t=1, 2, \ldots n$. 
Real world payoffs $P_r(.)$ are expressed in a simplified way as 

\begin{equation}
P_r(n)=P(0)+\sum_{k\leq N} g(x_t),	
\end{equation}
where  $g_t:\mathbb{R} \rightarrow \mathbb{R}$ is a measurable function representing the payoff;  $g$ may be path dependent (to accommodate a survival condition), that is, it is a function of the preceding period $\tau<t$  or on the cumulative sum $\sum_{\tau\leq t} g(x_\tau)$ to  introduce an absorbing barrier, say, bankruptcy avoidance, in which case we write:

\begin{mdframed}
\smallskip
\begin{equation}
	P^{(r)}(T)=P^{(r)}(0)+\sum_{t\leq n} \mathbbm{1}_{\left(\sum_{\tau < t}g(x_\tau)>b\right)}\; g(x_t),\label{ergodicity}
\end{equation}
	
\end{mdframed}
where $b$ is any arbitrary number in $\mathbb{R}$ that we call the survival mark and $\mathbbm{1}_{(.)}$ an indicator function $\in \{0,1\}$.

The last condition from the indicator function in Eq. \ref{ergodicity} is meant to handle ergodicity or lack of it \cite{taleb2016incerto}. 
\begin{commentary}
	P/L tautologically corresponds to the real world distribution, with an absorbing barrier at the survival condition.

\end{commentary}

\subsubsection{Frequency Space, }
The standard psychology literature has two approaches.
\smallskip
\subsubsection*{A--When tallying forecasts as a counter}
\begin{mdframed}
\smallskip
\begin{equation}
	P^{(p)}(n)=\frac{1}{n}\sum_{i\leq n} \mathbbm{1}_{X_t \in \chi},
\end{equation}
\end{mdframed}
where $\mathbbm{1}_{X_t \in \chi} \in \{0,1\}$ is an indicator that the random variable $x$ $\in \chi_t$ in in the "forecast range", and $T$ the total number of such forecasting events. 

\subsubsection*{B--When dealing with a score (calibration method)}
  in the absence of a visible net performance,  researchers produce some more advanced metric or score to measure calibration.  We select below the "gold standard", De Finetti's Brier score(DeFinetti, \cite{de2008philosophical}). It is favored since it doesn't allow arbitrage and requires perfect probabilistic calibration: someone betting than an event has a probability $1$ of occurring will get a perfect score only if the event occurs all the time. Let $f_t \in [0,1]$ be the probability announced by the forecaster for event $t$, 
\begin{mdframed}
\smallskip	
 \begin{equation}
 	\lambda^{(B)}_n = \frac{1}{n}\sum\limits _{t \leq n}(f_t-\mathbbm{1}_{X_t \in \chi})^2.
 \end{equation}
 \end{mdframed}
which needs to be minimized for a perfect probability assessor.

\subsubsection{Applications: M4 and M5 Competitions}
The M series (Makridakis \cite{makridakis2018m4}) evaluate forecasters using various methods to predict a point estimate (along with a range of possible values). The last competition in 2018, M4, largely relied on a series of scores,  $\lambda^{M4_j}$, which works well in situations where one has to forecast the first moment of the distribution and the dispersion around it. 

\begin{definition}[The M4 first moment forecasting scores]\label{M4def}
	The M4 competition precision score (Makridakis et al. \cite{makridakis2018m4}) judges competitors on the following metrics indexed by $j=1,2$
\begin{equation}
	\lambda^{(M4)_j}_{n}=\frac{1}{n}\sum_i^n \frac{\left|X_{f_i}-X_{r_i} \right|}{s_j}
\end{equation}
where $s_1=\frac{1}{2}\left(|X_{f_i}|+|X_{r_i}|\right)$ and $s_2$ is (usually) the raw mean absolute deviation for the observations available up to period $i$ (i.e., the mean absolute error from either "naive" forecasting or that from in sample tests), $X_{f_i}$ is the forecast for variable $i$ as a point estimate, $X_{r_i}$ is the realized variable and $n$ the number of experiments under scrutiny.
\end{definition}
In other word, it is an application of the Mean Absolute Scaled Error (MASE) and the symmetric Mean Absolute Percentage Error (sMAPE) \cite{hyndman2006another}.

The suggest M5 score (expected for 2020) adds the forecasts of extrema of the variables under considerations and repeats the same tests as the one for raw variables in Definition \ref{M4def}.

\subsection{Deriving Distributions}
\subsubsection{Distribution of $P^{(p)}(n)$}
\begin{remark}
	The average tally of the binary forecast (expressed as "higher" or "lower" than a threshold), $P^{(p)}(n)$ is, (using $p$ as a shortcut), asymptotically normal with mean $p$ and standard deviation $\sqrt{\frac{1}{n} (p - p^2)}$ regardless of the distribution class of the random variable $X$.
\end{remark}
The results are quite standard, but see appendix for the re-derivations.
 \smallskip
 
 \subsubsection{Distribution of the Brier Score $\lambda_n$}
 \begin{theorem}
Regardless of the distribution of the random variable $X$, without even assuming independence of $(f_1-\mathbbm{1}_A{_1}), \ldots, (f_n-\mathbbm{1}_A{_n})$, for $n<+\infty$, the score $\lambda_n$ has all moments of order $q$, $ \mathbbm{E}(\lambda_n^q)<+\infty$.
\end{theorem}
\begin{proof}
For all $i$, $(f_i -\mathbbm{1}_A{_i})^2 \leq 1$ .
\end{proof}

We can get actually closer to a full distribution of the score across independent betting policies. Assume binary predictions $f_i$ are independent and follow a beta distribution $\mathcal{B}(a,b)$ (which approximates or includes all unimodal distributions in $[0,1]$ (plus a Bernoulli via two Dirac functions), and let $p$ be the rate of success $p=\mathbbm{E}\left( \mathbbm{1}_A{_i}\right) $, the characteristic function of $\lambda_n$ for $n$ evaluations of the Brier score is
\begin{dmath}	
	\varphi_n(t)=\pi ^{n/2} \left(2^{-a-b+1} \Gamma (a+b) \\
	\left(p \, _2\tilde{F}_2\left(\frac{b+1}{2},\frac{b}{2};\frac{a+b}{2},\frac{1}{2} (a+b+1);\frac{i t}{n}\right)\\
	-(p-1) \, _2\tilde{F}_2\left(\frac{a+1}{2},\frac{a}{2};\frac{a+b}{2},\frac{1}{2} (a+b+1);\frac{i t}{n}\right)\right)\right)\label{brierdist}
\end{dmath}

Here $_2\tilde{F}_2$ is the generalized hypergeometric function regularized
$_2\tilde{F}_2(.,.;.,.;.)=$
$\frac{_2F_2(a;b;z)}{\left(\Gamma (b_1)\ldots \Gamma (b_q)\right)}$ and
$_pF_q(a;b;z)$ has series expansion $\sum _{k=0}^{\infty } \frac{ (a_1)_k \ldots  (a_p)_k}{ (b_1)_k \ldots  (b_p)_k} z^k/k!$, were $(a)_{(.)}$ is the Pockhammer symbol.



Hence we can prove the following: under the conditions of independence of the summands stated above,
\begin{mdframed}
\begin{equation}
	\lambda_n	\xrightarrow{D} \mathcal{N}\left(\mu,\sigma_n\right)
\end{equation}
	\end{mdframed}
where $\mathcal{N}$ denotes the Gaussian distribution with for first argument the mean and for second argument the standard deviation.

The proof and parametrization of $\mu$ and $\sigma_n$ are in the appendix.

\subsubsection{Distribution of the economic P/L or quantitative measure $P_r$}
\begin{remark}
Conditional on survival to time $T$, the distribution of the quantitative measure $P^{(r)}(T)$ will follow the distribution of the underlying variable $g(x)$.	
\end{remark}

The discussion is straightforward if there is no absorbing barrier (i.e., no survival condition).
\subsubsection{Distribution of the M4 score}
The distribution of an absolute deviation is in the same probability class as the variable itself. Thee Brier score is in the norm L2 and is based on the second moment (which always exists) as De Finetti has shown that it is more efficient to just a probability in square deviations. However for nonbinaries, it is vastly more efficient under fat tails to rely on absolute deviations, even when the second moment exists \cite{taleb2019statistical}.

 \section{Non-Verbalistic Payoff Functions and The Good News from Machine Learning}\label{ml}
 Earlier examples focused on simple payoff functions, with some cases where the conflation $I_1$ and $I_2$ can be benign (under the condition of being in a thin tailed environment). 
 
\subsubsection*{Inseparability of probability under nonlinear payoff function}

Now when we introduce a payoff function $g(.)$ that is nonlinear, that is that the economic or other quantifiable response to the random variable $X$ varies with the levels of $X$, the discrepancy becomes greater and the conflation worse.
\begin{commentary} [Probability as an integration kernel]
 Probability is just a kernel inside an integral or a summation, not a real thing on its own. The  economic world is about quantitative payoffs.
\end{commentary} 

\begin{remark}[Inseparability of probability]
Let $F: \mathcal{A} \rightarrow [0,1]$ be a probability distribution (with derivative $f$) and $g: \mathbbm{R} \rightarrow \mathbbm{R} $ a measurable function, the "payoff"". Clearly, for $\mathcal{A}'$ a subset of $\mathcal{A}$:
\begin{dmath*}
	\int_{\mathcal{A}'} g(x) \mathrm{d}F(x)=	\int_{\mathcal{A}'} f(x) g(x) \mathrm{d}x	 \neq \int_{\mathcal{A}'} f(x) \mathrm{d}x\;	 g\left(\int_{\mathcal{A}'} \mathrm{d}x\right)
	\end{dmath*}

In discrete terms, with $\pi(.)$ a probability mass function:
\begin{dmath}
	\sum_{x\in \mathcal{A}'} \pi(x) g(x) \neq \sum_{x\in \mathcal{A}'} \pi(x) \; g\left(\frac{1}{n} \sum_{x\in \mathcal{A}' }x \right) \\
	=\text{probability of event } \times \text{payoff of average event }
\end{dmath}
\end{remark}
\begin{proof}
Immediate by Jensen's inequality.	
\end{proof}

In other words, the probability of an event is an expected payoff only when, as we saw earlier,  $g(x)$ is a Heaviside theta function.

Next we focus on functions tractable mathematically or legally but not reliable verbalistically via "beliefs" or "predictions".
\subsubsection{Misunderstanding $g$}
Figure \ref{morgantrade} showing the mishedging story of Morgan Stanley is illustrative of verbalistic notions such as "collapse" mis-expressed in nonlinear exposures. In 2007 the Wall Street firm Morgan Stanley decided to "hedge" against a real estate "collapse", before the market in real estate started declining. The problem is that they didn't realize that "collapse" could take many values, some worse than they expected, and set themselves up 
to benefit if there were a mild decline, but lose much if there is a larger one. They ended up right in predicting the crisis, but lose $\$10$ billion from the "hedge".

Figure \ref{butterfly} shows a more complicated payoff, dubbed a "butterfly".

\subsubsection{$g$ and machine learning}

	We note that $g$ maps to various machine learning functions that produce exhaustive nonlinearities via the universal universal approximation theorem (Cybenko \cite{cybenko1989approximation}), or the generalized option payoff decompositions (see \textit{Dynamic Hedging}  \cite{taleb1997dynamic}).\footnote{We note the machine learning use of cross-entropy to gauge the difference between the distribution of a variable and the forecast, a technique that in effect captures the nonlinearities of $g(.)$. In effect on Principal Component Maps, distances using cross entropy are more effectively shown as the L^2 norm boosts extremes. For a longer discussion, \cite{taleb2019statistical}.}



Consider the function $\rho: (-\infty, \infty) \to [K,\infty)$, with $K$, the r.v. $X  \in \mathbb{R}$:
\begin{equation}
\rho_{K,p}(x)=k+\frac{\log \left(e^{p (x-K)}+1\right)}{p}	\label{rho}
\end{equation}
We can express all nonlinear payoff functions $g$ as, with the weighting $\omega_i \in \mathbbm{R}$: 
\begin{equation}
	g(x)=\sum_i \omega_i\; \rho_{{K_i},p}(x) 
\end{equation}
by some similarity, $\rho_{K,p}(x)$ maps to the value a call price with strike $K$ and time $t$ to expiration normalized to $1$, all rates set at $0$, with sole other parameter $\sigma$ the standard deviation of the underlying.

We note that the expectation of $g(.)$ is the sum of expectation of the ReLu functions:
\begin{equation}
	\mathbbm{E}\left(g(x)\right)=\sum_i \omega_i\; \mathbbm{E}\left(\rho_{{K_i},p}(x) \right)
\end{equation}
 The variance and other higher order statistical measurements are harder to obtain in closed or simple form.
 
 \begin{figure}
\includegraphics[width=\linewidth]{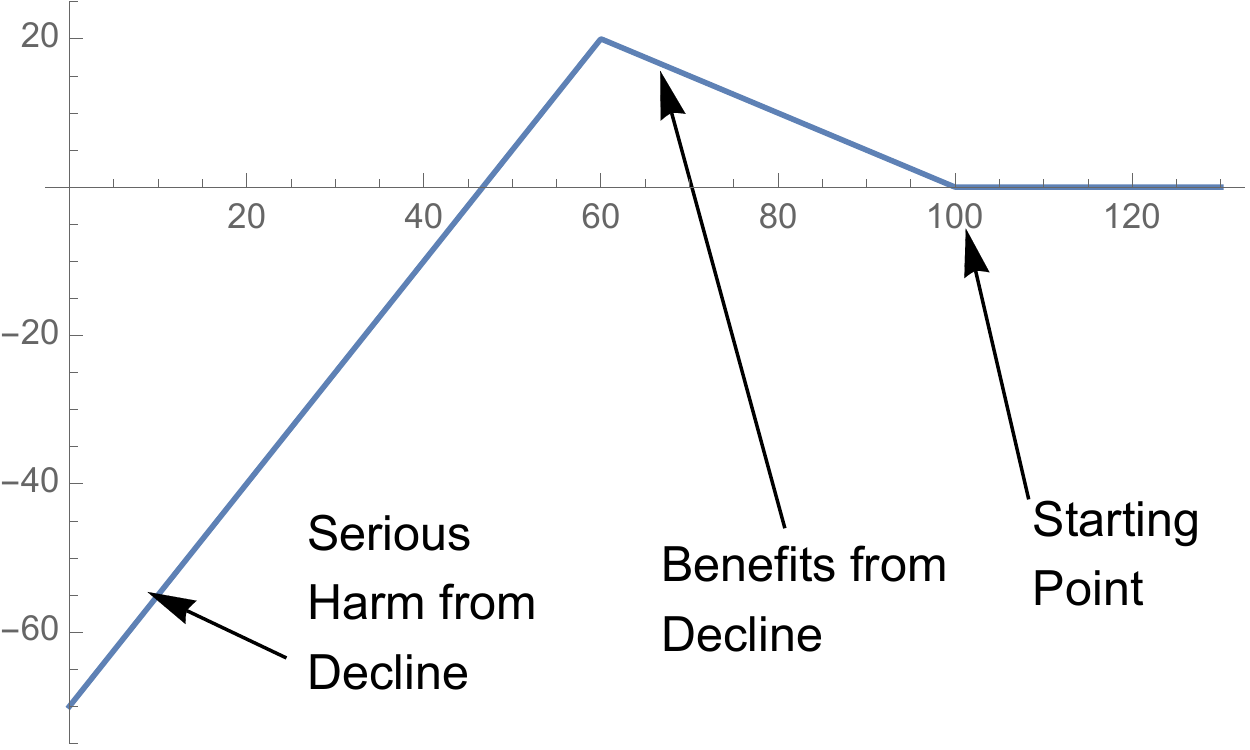}
\caption{The Morgan Stanley Story: an example of an elementary nonlinear payoff that cannot be described verbalistically.  The $x$ axis represents the variable, the vertical one the payoff.
This exposure is called in derivatives traders jargon a "Christmas Tree", achieved by purchasing a put with strike $K$ and selling a put with lower strike $K-\Delta_1$ and another with even lower strike $K-\Delta_2$, with $\Delta_2\geq \Delta_1\geq 0$.
}
\label{morgantrade}
\end{figure}

\begin{figure}
\includegraphics[width=\linewidth]{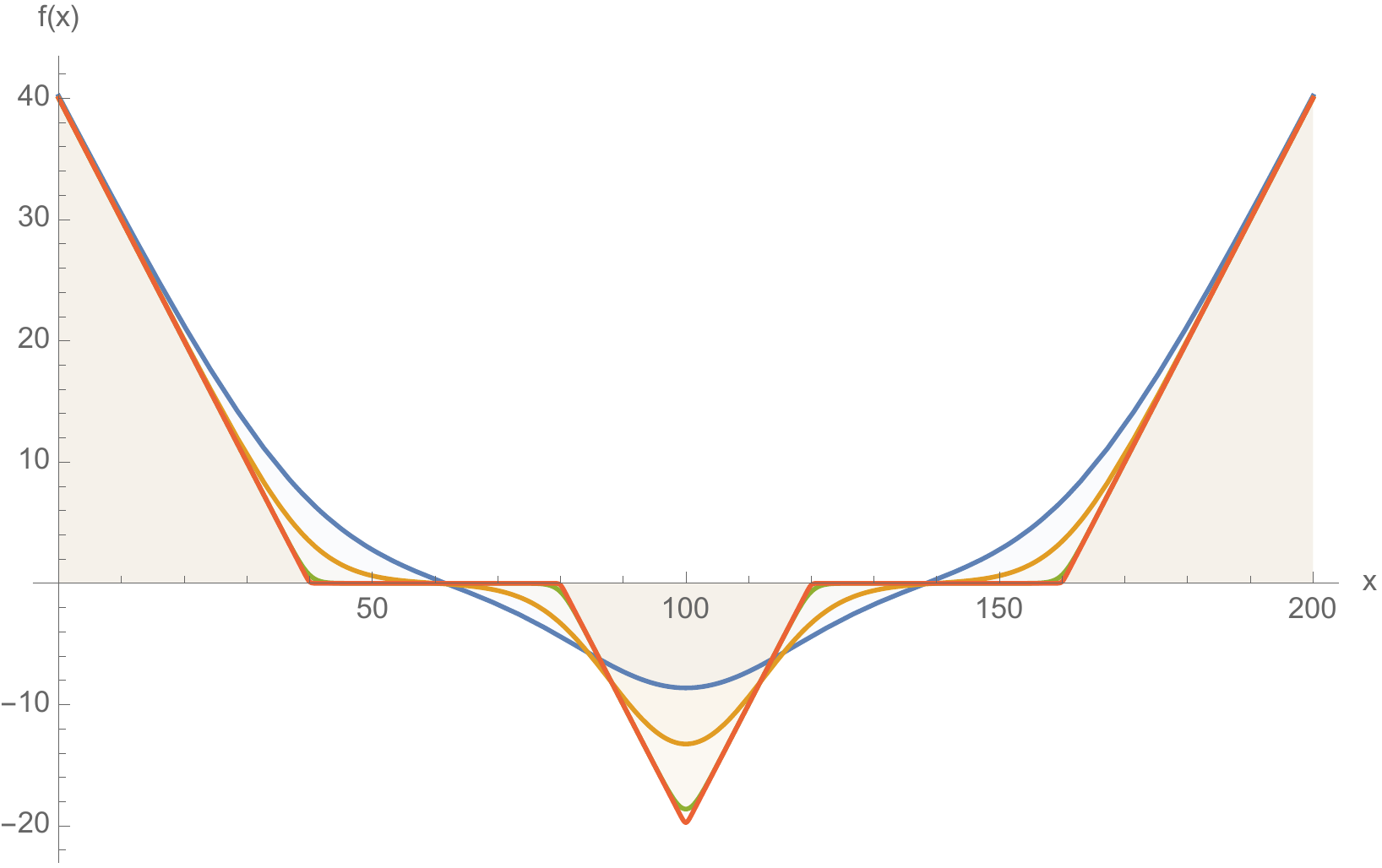}
\caption{A butterfly (built via a sum of options or ReLu $\rho_{K_i}$), with open tails on both sides and flipping first and second derivatives. Again, the $x$ axis represents the variable, the vertical one the payoff.
This example is particularly potent as it has no verbalistic correspondence but can be understood by option traders and machine learning.} \label{butterfly}
\end{figure}

\begin{commentary}
Risk management is about changing the payoff function $g(.)$	 rather than making "good forecasts".
\end{commentary}

We note that $g(.)$ is not a metric but a target to which one can apply various metrics. 

\begin{example}
	XIV was a volatility-linked exchange traded note that went bust, in February 2018, in the wake of a sudden jump in volatility, while having been able to "predict" such variable correctly. Simply, their $g(.)$ function was nonlinear not properly adapted to naive forecasting of the first moment of the random variable. \footnote{see \textit{Born to Die: Inside XIV, the Busted Volatility ETN}, \textit{Wall Street Journal}, Feb. 6, 2018).} The following simplified example can show us their error and the mismatch between correct forecasting and performance as well as  what happens under fat tails. 
	
	Consider the following "short volatility" payoff function,  $g(x)=1-x^2$ evaluated daily, where $x$ is a daily variation of an asset price, meaning if $x$ moves by up to $1$ unit (say, standard deviation), there is a profit, and losses beyond that. Such terms are found in a typical contract called "variance swap". Assume the note (or fund) are predicting that average daily (normalized) volatility will be "lower than 1". Now consider the following two types successions of deviations of $x$ for 7 days (expressed in normalized standard deviations). 

Succession 1 (thin tails): $\{1, 1, 1, 1, 1, 0, 0\}$. Mean variation= $0.71$. P/L: $g(x)=2$. Prediction is correct.

Succession 2 (fat tails): $\{0,0,0,0, 0,0,5 \}$. Mean variation= $0.71$ (same). P/L: $g(x)=-18$ (bust, really bust, but prediction is correct).

So in both cases a forecast of $<1$ is correct, but the lumping of the volatility --the fatness of tails-- made a huge difference.

This in a nutshell shows how, in the world of finance, "bad" forecasters can make great traders and decision makers, and vice versa.
\end{example}

\subsection*{Survival}
Decision making is sequential. Accordingly, miscalibration may be a good idea if it reduces the odds of being absorbed. See \cite{peters2014evaluating} and the appendix of \cite{taleb2016incerto}, which shows the difference between ensemble probability and time probability. The expectation of the sum of $n$ gamblers over a given day is different from that of a single gambler over $n$ days, owing to the conditioning.

In that sense, measuring the static performance of an agent who will eventually go bust (with probability one) is meaningless.\footnote {One suggestion by the author is for the M5 competition to correct for that by making "predictors" predict the minimum (or maximum) in a time series. This insures the predictor meets the constraint of staying above water.}

 \section{Conclusion:}
Finally,  that in the real world, it is the net performance (economic or other) that counts, and making "calibration" mistakes where it doesn't matter or can be helpful should be encouraged, not penalized. The bias variance argument is well known in machine learning \cite{hastie2009elements} as means to increase performance, in discussions of rationality (see \textit{Skin in the Game} \cite{taleb2016incerto}) as a necessary mechanism for survival, and  a very helpful psychological adaptation (Brighton and Gigerenzer \cite{brighton2012homo} show a potent argument that if it is a bias, it is a pretty useful one.) If a mistake doesn't cost you anything --or helps you survive or improve your outcomes-- it is clearly not a mistake. And if it costs you something, and has been present in society for a long time, consider that there may be hidden evolutionary advantages to these types of mistakes --of the following sort: \textbf{mistaking a bear for a stone} is worse than \textbf{mistaking a stone for a bear}. 

We have shown that, in risk management (and while forecasting socio-economic and financial variables), one should never operate in probability space.

\appendix
\subsection{Distribution of Binary Tally $P^{(p)}(n)$}
We are dealing with an average of Bernoulli random variables, with well known results but worth redoing.  The characteristic function of a Bernoulli distribution with parameter $p$ is $\psi(t)=1 - p + E^(I t) p$. We are concerned with the $N$-summed cumulant generating function $\psi'(\omega)=\log \psi(\frac{\omega}{N})^N$.
We have $\kappa(p)$ the cumulant of order $p$:
$$\kappa(p)=-i^p \frac{\partial ^p\psi'}{\partial t^p}\bigg|_{t\to 0}$$
So: $\kappa(1)=p,$ 
$\kappa(2)=\frac{(1-p) p}{N},$ 
$\kappa(3)=\frac{(p-1) p (2 p-1)}{N^2},$
$\kappa(4)=\frac{(1-p) p (6 (p-1) p+1)}{N^3}$, which proves that $P^{(p)}(N)$ converges by the law of large numbers at speed $\sqrt{N}$, and by the central limit theorem arrives to the Gaussian at a rate of $\frac{1}{N}$, (since from the cumulants above, its kurtosis = $3-\frac{6 (p-1) p+1}{n (p-1) p}$).

\subsection{Distribution of the Brier Score}
\subsubsection*{Base probability $f$}
First, we consider the distribution of $f$ the base probability. We use a beta distribution that covers both the conditional and unconditional case (it is a matter of parametrization of $a$ and $b$ in Eq. \ref{brierdist}).
\subsubsection*{Distribution of the probability}
 Let us refresh a standard result behind nonparametric discussions and tests, dating from Kolmogorov \cite{kolmogorov1933sulla} to show the rationale behind the claim that the probability distribution of probability (sic) is robust --in other words the distribution of the probability of $X$ doesn't depend on the distribution of $X$, 
 (\cite{dietrich2017uncertainty} \cite{keren1991calibration}). 

The probability integral transform is as follows. Let $X$ have a continuous distribution for which the cumulative distribution function (CDF) is $F_X$. Then --in the absence of additional information --the random variable $U$ defined as $U=F_X(X)$ is uniform between $0$ and $1$.
The proof is as follows: For $t \in [0,1]$,
\begin{multline}
	\mathbb{P}(Y \leq u) =\\
	 P(F_X(X) \leq u)= P(X \leq F_X^{-1}(u)) = F_X(F_X^{-1}(u)) = u	
	 \end{multline}
which is the cumulative distribution function of the uniform. 
This is the case regardless of the probability distribution of $X$.

%
%

Clearly we are dealing with 1) $f$ beta distributed (either as a special case the uniform distribution when purely random, as derived above, or a beta distribution when one has some accuracy, for which the uniform is a special case), 
and 2) $\mathbbm{1}_A{_t}$ a Bernoulli variable with probability $p$.

Let us consider the general case. Let $g_{a,b}$ be the PDF of the Beta distribution:

$$g_{a,b}(x)= \frac{x^{a-1} (1-x)^{b-1}}{B(a,b)}, \;  0<x<1.$$

 The results, a bit unwieldy but controllable:
$$\mu=\frac{\left(a^2 (-(p-1))-a p+a+b (b+1) p\right) \Gamma (a+b)}{\Gamma (a+b+2)},$$

\begin{dmath*}
	\sigma_n^2=-\frac{1}{n (a+b)^2 (a+b+1)^2}
\left(a^2 (p-1)+a (p-1)-b (b+1) p\right)^2+\frac{1}{(a+b+2) (a+b+3)}(a+b) (a+b+1) (p (a-b) (a+b+3) (a (a+3)+(b+1) (b+2))-a (a+1) (a+2) (a+3)).
\end{dmath*}

We can further verify that the Brier score has thinner tails than the Gaussian as its kurtosis is lower than 3.

\begin{proof}

We start with $y_j= (f-\mathbbm{1}_A{_j})$, the difference between a continuous Beta distributed random variable and a discrete Bernoulli one, both indexed by $j$. The characteristic function of $y_j$, $\Psi^{(y)}_f= \left(1+p \left(-1+e^{-i t}\right)\right) \, _1F_1(a;a+b;i t)$ where $_1F_1(.;.;.)$ is the hypergeometric distribution
$_1F_1(a;b;z)=\sum _{k=0}^{\infty } \frac{a_k \frac{z^k}{k}!}{b_k}$. 

From here we get the characteristic function for $y_j^2=(f_j-\mathbbm{1}_A{_j})^2$:
\begin{dmath}
\Psi^{(y^2)}(t)=\sqrt{\pi } 2^{-a-b+1} \Gamma (a+b) \left(p \, _2\tilde{F}_2\left(\frac{b+1}{2},\frac{b}{2};\frac{a+b}{2},\frac{1}{2} (a+b+1);i t\right)-(p-1) \, _2\tilde{F}_2\left(\frac{a+1}{2},\frac{a}{2};\frac{a+b}{2},\frac{1}{2} (a+b+1);i t\right)\right),
\end{dmath}
where $_2\tilde{F}_2$ is the generalized hypergeometric function regularized
$_2\tilde{F}_2(.,.;.,.;.)=$
$\frac{_2F_2(a;b;z)}{\left(\Gamma (b_1)\ldots \Gamma (b_q)\right)}$ and
$_pF_q(a;b;z)$ has series expansion $\sum _{k=0}^{\infty } \frac{ (a_1)_k \ldots  (a_p)_k}{ (b_1)_k \ldots  (b_p)_k} z^k/k!$, were $(a)_{(.)}$ is the Pockhammer symbol.

We can proceed to prove directly from there the convergence in distribution for the average $\frac{1}{n}\sum_i^n y_i^2$:
\begin{dmath}
	\underset{n\to \infty }{\text{lim}}\Psi_{y^2}(t/n)^n=\\
	\exp \left(-\frac{i t (p (a-b) (a+b+1)-a (a+1))}{(a+b) (a+b+1)}\right),
\end{dmath}
which is that of a degenerate Gaussian (Dirac) with location parameter $\frac{p (b-a)+\frac{a (a+1)}{a+b+1}}{a+b}$.


We can finally assess the speed of convergence, the rate at which higher moments map to those of a Gaussian distribution: consider the  behavior of the $4^{th}$ cumulant $\kappa_4= -i \frac{\partial ^4\log \Psi_.(.)}{\partial t^4}|_ {t\to 0}$:

1) in the maximum entropy case of $a=b=1$:
 $$\kappa_4|_{a=1,b=1}=-\frac{6}{7 n},$$
 regardless of $p$.
 
2) In the maximum variance case, using l'H\^opital:
$$\lim_{\substack{a\to 0 \\ b \to 0}} \kappa_4=-\frac{6 (p-1) p+1}{n (p-1) p}.$$	
 Se we have $ \frac{\kappa_4}{\kappa_2^2} \underset{n\to \infty}{ \rightarrow 0}$ at rate $n^{-1}$.
\end{proof}

Further, we can extract its probability density function of the Brier score for $N=1$: for $0<z<1$,
\begin{dmath}
	p(z)=\frac{\Gamma (a+b) \left((p-1) z^{a/2} \left(1-\sqrt{z}\right)^b-p \left(1-\sqrt{z}\right)^a z^{b/2}\right)}{2 \left(\sqrt{z}-1\right) z \Gamma (a) \Gamma (b)}.
\end{dmath}

\bibliographystyle{IEEEtran}
\bibliography{/Users/nntaleb/Dropbox/Central-bibliography}

\end{document}